\documentclass[11pt]{article}

\usepackage[a4paper,margin=1in]{geometry}
\usepackage[T1]{fontenc}
\usepackage[utf8]{inputenc}
\usepackage{lmodern}
\usepackage{microtype}
\usepackage{amsmath,amssymb,amsthm,mathtools,bm}
\usepackage{algorithm}
\usepackage{algpseudocode}
\usepackage{graphicx}
\usepackage{booktabs}
\usepackage{array}
\usepackage{enumitem}
\usepackage{siunitx}
\usepackage{xcolor}
\usepackage{hyperref}
\usepackage[nameinlink,noabbrev]{cleveref}
\usepackage{tabularx}
\usepackage{float}
\usepackage{placeins}
\usepackage{multirow}
\usepackage{longtable}

\hypersetup{
    colorlinks=true,
    linkcolor=black,
    citecolor=black,
    urlcolor=black
}
\setlist[itemize]{leftmargin=1.6em,itemsep=2pt,topsep=3pt}

\newcommand{\E}{\mathbb{E}}
\newcommand{\Prob}{\mathbb{P}}
\newcommand{\R}{\mathbb{R}}

\newcommand{\CN}{\mathcal{CN}}
\newcommand{\norm}[1]{\left\lVert #1\right\rVert}
\newcommand{\abs}[1]{\left\lvert #1\right\rvert}
\newcommand{\vect}[1]{\bm{#1}}

\DeclareMathOperator*{\argmin}{arg\,min}

\newtheorem{theorem}{Theorem}
\newtheorem{proposition}{Proposition}
\newtheorem{remark}{Remark}

\newcommand{\safeincludegraphics}[2][]{%
  \IfFileExists{#2}{\includegraphics[#1]{#2}}{%
    \fbox{\parbox[c][0.25\textheight][c]{0.9\linewidth}{\centering Missing figure file\\\texttt{\detokenize{#2}}}}%
  }%
}

\title{A Baseline Mobility-Aware IRS-Assisted Uplink Framework\\With Energy-Detection-Based Channel Allocation}
\author{Ardavan Rahimian\\School of Engineering, Ulster University, Belfast, U.K.\\\texttt{a.rahimian@ulster.ac.uk}}
\date{}

\begin{document}
\maketitle

\begin{abstract}
This paper develops a self-contained framework for studying a mobility-aware intelligent reflecting surface (IRS)-assisted multi-node uplink under simplified but explicit modeling assumptions. The considered system combines direct and IRS-assisted narrowband propagation, geometric IRS phase control with finite-bit phase quantization, adaptive IRS-user focusing based on inverse-rate priority weights, and sequential channel allocation guided by energy detection. The analytical development is restricted to a physics-based two-hop cascaded path-loss formulation with appropriate scaling, an expectation-level reflected-power characterization under the stated independence assumptions, and the exact chi-square threshold for energy detection, together with its large-sample Gaussian approximation. A MATLAB  implementation is used to generate a sample run, which is interpreted as a numerical example. This work is intended as a consistent, practically-aligned baseline to support future extensions involving richer mobility models or more advanced scheduling policies.
\end{abstract}


\section{Introduction}
Intelligent reflecting surfaces (IRSs) or programmable metasurfaces provide advanced mechanisms and methods for shaping radio propagation by adjusting the phases of many reflecting elements \cite{WuTWC2019,WuCommMag2020,DiRenzoJSAC2020,WuTutorialTCom2021,BjornsonSPM2022}. Related studies have also examined physics-based path-loss modeling, mobility-oriented IRS behavior, and comparisons with other relaying paradigms \cite{OzdoganWCL2020,WeiTVT2025,BjornsonWCL2020}. In multi-user uplinks, IRS control interacts with user prioritization, interference, and channel access. Under mobility, even a simplified model must track changing geometry, time-varying fading, and per-slot scheduling choices, as well as the associated factors in such distributed networks.

The purpose of this work is to be practical and to serve as a baseline implementation of such subsystems for low-power wireless networks or Internet of Things (IoT) scenarios. It establishes a coherent foundation, retaining only the segments and associated analyses that are most suitable for such use cases. Within this scope, the framework combines the following: a narrowband uplink model with direct and IRS-assisted propagation; geometric IRS phase alignment with finite-bit quantization; inverse-rate priority-weighted random selection of the IRS focus user; and sequential channel allocation driven by an exact energy-detection threshold.

The key contributions of this proposed framework are as follows: it adopts a physics-based two-hop IRS path-loss model with the correct $L_0^2$ cascaded scaling \cite{OzdoganWCL2020}; it states a mean reflected-power proposition under explicit independence assumptions; it derives and uses the exact chi-square threshold for energy detection, together with the standard large-sample Gaussian approximation \cite{Urkowitz1967,Digham2007}; it formalizes the implemented adaptive IRS-focus policy and the sensing-guided sequential channel-allocation rule; and it reports one seeded numerical example as an illustrative run only, while separating numerical output from theoretical claims.

The work first formulates a simplified mobility-aware IRS-assisted uplink model and defines the associated notation and propagation assumptions. It then develops the analytical results that are justified within this scope: the direct and cascaded path-loss structure, an expectation-level characterization of the reflected power under independent small-scale fading, and the exact energy-detection threshold used for sensing-guided channel allocation. On top of this model, the paper formalizes the implemented inverse-rate adaptive focusing rule and the sequential channel-assignment procedure, and finally reports a seeded numerical example generated by the MATLAB simulator to illustrate the behavior of the framework. In this way, the paper aims to provide a foundational reference that is suitable for subsequent refinement and extension.

\section*{Notation}
\begin{table}[H]
\centering
\caption{Main notation used in the paper.}
\label{tab:notation}
\small
\begin{tabularx}{\linewidth}{@{}lX@{}}
\toprule
Symbol & Meaning \\
\midrule
$K$ & Number of uplink nodes/users. \\
$N=N_xN_y$ & Number of IRS elements, with $N_x$ and $N_y$ along the two array dimensions. \\
$\vect{u}_k(t)$, $\vect{b}$, $\vect{r}_0$, $\vect{r}_n$ & Position of node $k$, base station (BS), IRS center, and IRS element $n$, respectively. \\
$\lambda$, $f_c$, $c$ & Wavelength, carrier frequency, and speed of light. \\
$\rho$ & IRS reflection efficiency. \\
$h_{k,\mathrm d}(t)$, $h_{k,\mathrm{IRS}}(t)$, $h_k(t)$ & Direct, IRS-assisted, and total uplink channel of user $k$. \\
$\beta_{\mathrm d}(d)$ & Large-scale direct-link gain at distance $d$. \\
$\beta_{12}(d_1,d_2)$ & Two-hop cascaded large-scale gain for user-to-IRS distance $d_1$ and IRS-to-BS distance $d_2$. \\
$\psi_n(t)$ & IRS phase applied by element $n$ at slot $t$. \\
$\vect{v}_k(t)$, $f_D$ & Node velocity and corresponding Doppler shift. \\
$P_{\mathrm{tx}}$, $\sigma^2$ & User transmit power and receiver noise power. \\
$c_k(t)$, $C$ & Channel assigned to user $k$ and total number of orthogonal channels. \\
$M$, $P_{\mathrm{fa}}$, $\gamma$ & Number of sensing samples, target false-alarm probability, and energy-detection threshold. \\
$R_k(t)$, $\bar R_k(t)$ & Instantaneous rate and sliding-window average rate of user $k$. \\
$W$, $\beta$, $\epsilon$ & Averaging-window length, priority exponent, and small regularization constant in the scheduler. \\
\bottomrule
\end{tabularx}
\end{table}

\section{System Model and Channel Characterization}
Consider an uplink network with $K$ single-antenna devices and a single-antenna BS. The BS is located at $\vect{b}\in\R^3$. An IRS with $N=N_xN_y$ reflecting elements is centered at $\vect{r}_0\in\R^3$ and its $n$th element is at position $\vect{r}_n$. Device $k$ is at position $\vect{u}_k(t)\in\R^3$ at slot $t$. Devices move inside a bounded region $\mathcal{A}\subset\R^3$ with reflective boundary conditions.

\subsection{Composite uplink channel}
The complex baseband uplink channel from node $k$ to the BS is modeled as
\begin{equation}
    h_k(t)=h_{k,\mathrm d}(t)+h_{k,\mathrm{IRS}}(t),
    \label{eq:htot}
\end{equation}
where $h_{k,\mathrm d}(t)$ is the direct path and $h_{k,\mathrm{IRS}}(t)$ is the IRS-assisted component.

The direct channel is written as
\begin{equation}
    h_{k,\mathrm d}(t)=\sqrt{\beta_{\mathrm d}(d_{kb}(t))}\,\tilde h_{k,\mathrm d}(t)e^{-j\frac{2\pi}{\lambda}d_{kb}(t)},
    \label{eq:hdir}
\end{equation}
where
\begin{equation}
    d_{kb}(t)=\norm{\vect b-\vect u_k(t)}
\end{equation}
and $\tilde h_{k,\mathrm d}(t)\sim\CN(0,1)$ in the simulator.

The IRS-assisted component is modeled as
\begin{equation}
    h_{k,\mathrm{IRS}}(t)=\sum_{n=1}^{N}\rho\,g_{k\to n}(t)g_{n\to b}e^{j\psi_n(t)},
    \label{eq:hirs}
\end{equation}
where $\rho\in[0,1]$ is the reflection efficiency and $\psi_n(t)$ is the phase shift of element $n$.

\subsection{Two-hop cascaded channel model}
The user-to-element and element-to-BS coefficients are modeled as
\begin{align}
    g_{k\to n}(t) &= \sqrt{\beta_{\mathrm d}(d_{kn}(t))}\,\tilde g_{k,n}(t)e^{-j\frac{2\pi}{\lambda}d_{kn}(t)}, \label{eq:gkn}\\
    g_{n\to b} &= \sqrt{\beta_{\mathrm d}(d_{nb})}\,\tilde g_{n,b}e^{-j\frac{2\pi}{\lambda}d_{nb}}, \label{eq:gnb}
\end{align}
with
\begin{equation}
    d_{kn}(t)=\norm{\vect r_n-\vect u_k(t)},
    \qquad
    d_{nb}=\norm{\vect b-\vect r_n}.
\end{equation}
The BS and the IRS are fixed, so the element-to-BS link is treated as quasi-static over the simulated run. Therefore, the cascaded large-scale factor is
\begin{equation}
    \beta_{12}(d_{kn}(t),d_{nb})=\beta_{\mathrm d}(d_{kn}(t))\beta_{\mathrm d}(d_{nb}).
    \label{eq:beta12def}
\end{equation}
Substituting \cref{eq:gkn,eq:gnb} into \cref{eq:hirs} yields
\begin{equation}
    h_{k,\mathrm{IRS}}(t)=\rho\sum_{n=1}^{N}\sqrt{\beta_{12}(d_{kn}(t),d_{nb})}\,\tilde g_{k,n}(t)\tilde g_{n,b}
    e^{-j\frac{2\pi}{\lambda}(d_{kn}(t)+d_{nb})}e^{j\psi_n(t)}.
    \label{eq:hirs_expanded}
\end{equation}
This is also the form implemented in the code before summation over the IRS elements.

\subsection{Physics-based path loss}
The direct and cascaded large-scale gains are modeled as
\begin{align}
    \beta_{\mathrm d}(d) &=
    \begin{cases}
        1, & d\le d_0,\\
        (L_0d^{\alpha})^{-1}, & d>d_0,
    \end{cases}
    \label{eq:pld}\\
    \beta_{12}(d_1,d_2) &=
    \begin{cases}
        1, & d_1d_2\le d_0^2,\\
        (L_0^2d_1^{\alpha}d_2^{\alpha})^{-1}, & d_1d_2>d_0^2,
    \end{cases}
    \label{eq:pl12}
\end{align}
where
\begin{equation}
    L_0=\left(\frac{4\pi}{\lambda}\right)^2,
    \qquad
    d_0=\frac{\lambda}{2\pi},
\end{equation}
and $\alpha$ is the path-loss exponent. The essential point is the two-hop scaling factor $L_0^2$, which is consistent with the physics-based cascaded path-loss model in \cite{OzdoganWCL2020}.

\subsection{IRS phase design}
For a designated focus user $k^{\star}(t)$, the simulator uses geometric phase alignment,
\begin{equation}
    \psi_n^{\mathrm{geom}}(t)=\mathrm{wrap}_{[-\pi,\pi)}\!\left(\frac{2\pi}{\lambda}\bigl(d_{k^{\star}n}(t)+d_{nb}\bigr)\right),
    \label{eq:psigeom}
\end{equation}
which compensates the deterministic path-length phase term in \cref{eq:hirs_expanded}. The code also supports an idealized perfect channel state information (CSI)-based phase-design mode,
\begin{equation}
    \psi_n^{\mathrm{CSI}}(t)=\mathrm{wrap}_{[-\pi,\pi)}\!\left(\frac{2\pi}{\lambda}\bigl(d_{k^{\star}n}(t)+d_{nb}\bigr)-\arg \tilde g_{k^{\star},n}(t)-\arg \tilde g_{n,b}\right),
    \label{eq:psicsi}
\end{equation}
but the reported run uses geometric control only. After phase design, the applied IRS phases are quantized using $b$ bits. In the reported example, $b=3$, so each phase is mapped to one of $2^b=8$ uniformly spaced quantization levels.

\begin{proposition}[Mean reflected-power scaling]
Assume for a fixed user $k$ that, for each IRS element $n$, the small-scale coefficients
$\tilde g_{k,n}(t)$ and $\tilde g_{n,b}$ are zero-mean, unit-variance, and mutually independent. Also, assume that the pairs
$\{(\tilde g_{k,n}(t),\tilde g_{n,b})\}_{n=1}^{N}$
are independent across different IRS elements. Then
\begin{equation}
    \E\!\left[\abs{h_{k,\mathrm{IRS}}(t)}^2\right]
    =\rho^2\sum_{n=1}^{N}\beta_{12}(d_{kn}(t),d_{nb}).
    \label{eq:meanscaling}
\end{equation}
Consequently, if the per-element large-scale factors remain of comparable order, the mean reflected power scales linearly with $N$, i.e., as $\mathcal{O}(N)$.
\end{proposition}

\begin{proof}
Define
\begin{equation}
    X_n(t)=\tilde g_{k,n}(t)\tilde g_{n,b}
\end{equation}
and
\begin{equation}
    a_n(t)=\rho\sqrt{\beta_{12}(d_{kn}(t),d_{nb})}
    e^{-j\frac{2\pi}{\lambda}(d_{kn}(t)+d_{nb})}e^{j\psi_n(t)}.
\end{equation}
Then \cref{eq:hirs_expanded} can be written as
\begin{equation}
    h_{k,\mathrm{IRS}}(t)=\sum_{n=1}^{N} a_n(t)X_n(t).
\end{equation}
Therefore,
\begin{align}
    \E\!\left[\abs{h_{k,\mathrm{IRS}}(t)}^2\right]
    &= \E\!\left[\left(\sum_{n=1}^{N} a_n(t)X_n(t)\right)
    \left(\sum_{m=1}^{N} a_m(t)X_m(t)\right)^{\!*}\right] \\
    &= \sum_{n=1}^{N}\sum_{m=1}^{N}
    a_n(t)a_m^\ast(t)\,\E\!\left[X_n(t)X_m^\ast(t)\right].
\end{align}
For \(n\neq m\), independence across IRS elements and the zero-mean assumptions imply
\begin{equation}
    \E[X_n(t)X_m^\ast(t)] = 0.
\end{equation}
For \(n=m\),
\begin{align}
    \E[\abs{X_n(t)}^2]
    &= \E[\abs{\tilde g_{k,n}(t)}^2]\,
       \E[\abs{\tilde g_{n,b}}^2] \\
    &= 1.
\end{align}
Hence, only the diagonal terms remain, and thus
\begin{equation}
    \E\!\left[\abs{h_{k,\mathrm{IRS}}(t)}^2\right]
    = \sum_{n=1}^{N}\abs{a_n(t)}^2
    = \rho^2 \sum_{n=1}^{N}\beta_{12}(d_{kn}(t),d_{nb}),
\end{equation}
which proves \cref{eq:meanscaling}.
\end{proof}

\begin{remark}
This work does not claim a deterministic instantaneous $\Theta(N^2)$ scaling law for the reflected power. Stronger coherent-sum scaling statements would require additional assumptions and are not needed for the present framework.
\end{remark}

\subsection{Simplified mobility and Doppler model}
The mobility model is intentionally simple. Each node is initialized with a random planar heading and a speed no larger than $v_{\max}$. At slot index $\ell$, the position evolves according to
\begin{equation}
    \vect u_k[\ell+1]=\vect u_k[\ell]+\Delta t\,\vect v_k[\ell],
    \label{eq:kin}
\end{equation}
followed by specular reflection if a boundary of $\mathcal A$ is crossed.

The small-scale fading evolution used in the code is hybrid. First, a coherence-time proxy is formed from the maximum Doppler shift,
\begin{equation}
    f_{D,\max}=\frac{\norm{\vect v_k[\ell]}}{\lambda},
    \qquad
    T_{\mathrm{coh}}\approx \frac{0.423}{f_{D,\max}},
    \label{eq:tcoh}
\end{equation}
with a small lower bound used in the implementation to avoid division by very small values. When the elapsed time since the last redraw exceeds this proxy, the direct-link coefficient and the user-to-IRS coefficients associated with that node are redrawn as fresh circularly symmetric complex Gaussian variables, while the IRS-to-BS coefficients are kept quasi-static over the run. This method is a lightweight abstraction rather than a full mobility model as in other studies \cite{WeiTVT2025}.

Between redraws, phase evolution is modeled through a Doppler shift. For a propagation direction represented by a unit vector $\hat{\vect k}$,
\begin{equation}
    f_D=\frac{\vect v_k[\ell]^\top \hat{\vect k}}{\lambda}
    =\frac{\norm{\vect v_k[\ell]}}{\lambda}\cos\theta,
    \label{eq:dopp}
\end{equation}
so the phase evolution of the corresponding small-scale coefficient is approximated by
\begin{equation}
    \tilde h(t+\Delta t)=\tilde h(t)e^{j2\pi f_D\Delta t}.
    \label{eq:evol}
\end{equation}
This is a simplified mobility-aware fading abstraction, not a full Jakes-type stochastic time-correlation model \cite{Sklar1997}.

\section{SINR Model and Energy-Based Channel Allocation}
At each slot, every user is assigned one of $C$ orthogonal channels. Users assigned to the same channel are modeled as mutual interferers in the signal-to-interference-plus-noise ratio (SINR) calculation. Let $c_k(t)\in\{1,\dots,C\}$ denote the channel assigned to user $k$ at slot $t$, and define the instantaneous received signal power as
\begin{equation}
    P_k^{\mathrm{rx}}(t)=P_{\mathrm{tx}}\abs{h_k(t)}^2.
    \label{eq:prx}
\end{equation}
Then the simulator computes the instantaneous SINR of user $k$ as
\begin{equation}
    \mathrm{SINR}_k(t)=\frac{P_k^{\mathrm{rx}}(t)}{\sum_{j\neq k:\,c_j(t)=c_k(t)} P_j^{\mathrm{rx}}(t)+\sigma^2},
    \label{eq:sinr}
\end{equation}
where $\sigma^2$ is the receiver noise power.

With bandwidth $B$, the instantaneous rate is
\begin{equation}
    R_k(t)=B\log_2\bigl(1+\mathrm{SINR}_k(t)\bigr),
    \label{eq:rate}
\end{equation}
and in the simulator this rate is set to zero whenever $\mathrm{SINR}_k(t)<\gamma_{\mathrm{dec,lin}}$, where
\begin{equation}
    \gamma_{\mathrm{dec,lin}}=10^{\gamma_{\mathrm{dec}}/10}.
\end{equation}
For the reported run, $\gamma_{\mathrm{dec}}=-10$ dB.

\subsection{Energy detector and exact threshold}
For a candidate channel $c$, the sensing statistic uses $M$ complex samples,
\begin{equation}
    T_c=\sum_{m=1}^{M}\abs{y_c[m]}^2.
    \label{eq:tc}
\end{equation}
Under the noise-only hypothesis $\mathcal H_0$, the code assumes that
\begin{equation}
    y_c[m]\sim \CN(0,\sigma^2), \qquad m=1,\dots,M,
\end{equation}
and that these samples are independent across $m$. Since $2\abs{y_c[m]}^2/\sigma^2\sim\chi_2^2$, summing over $M$ samples gives the standard result below.

\begin{theorem}[Exact energy-detection threshold]
Under $\mathcal H_0$, assume that
\begin{equation}
    y_c[m]\sim \CN(0,\sigma^2), \qquad m=1,\dots,M,
\end{equation}
and that the samples are independent across $m$. Then
\begin{equation}
    \frac{2}{\sigma^2}T_c\sim\chi^2_{2M}.
\end{equation}
Therefore, for a target false-alarm probability $P_{\mathrm{fa}}$,
\begin{equation}
    \gamma=\frac{\sigma^2}{2}\,F^{-1}_{\chi^2_{2M}}(1-P_{\mathrm{fa}})
    \label{eq:gammaexact}
\end{equation}
ensures
\begin{equation}
    \Prob(T_c>\gamma\mid \mathcal H_0)=P_{\mathrm{fa}}.
\end{equation}
For large $M$, the approximation
\begin{equation}
    \gamma\approx \sigma^2\left(M+\sqrt{M}\,z_{1-P_{\mathrm{fa}}}\right)
    \label{eq:gammanorm}
\end{equation}
is obtained from the Gaussian approximation to the chi-square law, where $z_{1-P_{\mathrm{fa}}}$ is the $(1-P_{\mathrm{fa}})$ standard-normal quantile.
\end{theorem}

\begin{proof}
Under $\mathcal H_0$, each sample satisfies $y_c[m]\sim\CN(0,\sigma^2)$, so its real and imaginary parts are independent $\mathcal N(0,\sigma^2/2)$ random variables. Hence
\begin{equation}
    \frac{2|y_c[m]|^2}{\sigma^2}\sim\chi_2^2.
\end{equation}
Assuming independence across $m$, summing over $M$ samples yields
\begin{equation}
    \frac{2}{\sigma^2}T_c=\sum_{m=1}^{M}\frac{2|y_c[m]|^2}{\sigma^2}\sim\chi_{2M}^2.
\end{equation}
The threshold in \cref{eq:gammaexact} follows by inversion of the cumulative distribution function of $\chi^2_{2M}$. For large $M$, the Gaussian approximation follows from the fact that a chi-square random variable with $2M$ degrees of freedom has mean $2M$ and variance $4M$. Accordingly, replacing $\chi^2_{2M}$ by its normal approximation and rescaling by $\sigma^2/2$ yields \cref{eq:gammanorm}.
\end{proof}

\subsection{Sequential channel allocation rule}
The code forms a running energy estimate for each channel before assigning users. Let $T_{c,\mathrm{noise}}(t)$ denote a noise-only energy draw for channel $c$, generated according to the same $\mathcal H_0$ model. Let $c_j(t-1)$ denote the channel used by user $j$ in the previous slot. Then the initial running energy of channel $c$ at slot $t$ is
\begin{equation}
    E_c^{(0)}(t)=T_{c,\mathrm{noise}}(t)+M\sum_{j:\,c_j(t-1)=c} P_j^{\mathrm{rx}}(t).
    \label{eq:initE}
\end{equation}
The first term models sensing noise, while the second term initializes the channel-energy estimate using the current-slot received-power values of users that occupied channel $c$ in the previous slot.

Users are then processed sequentially in node-index order. If user $k$ is currently being assigned and $E_c^{(k-1)}(t)$ denotes the running energy of channel $c$ after processing users $1,\dots,k-1$, the implemented assignment rule is
\begin{equation}
    c_k(t)=
    \begin{cases}
        \text{the first } c \text{ such that } E_c^{(k-1)}(t)<\gamma, & \text{if such a channel exists},\\
        \argmin_{1\le c\le C} E_c^{(k-1)}(t), & \text{otherwise},
    \end{cases}
    \label{eq:alloc}
\end{equation}
followed by the update
\begin{equation}
    E_{c_k(t)}^{(k)}(t)=E_{c_k(t)}^{(k-1)}(t)+M P_k^{\mathrm{rx}}(t).
    \label{eq:eupdate}
\end{equation}
This is a sequential energy-guided heuristic that is internally consistent but not claimed to be globally optimal.

\section{Adaptive IRS Focus Scheduling}
The simulator selects one focus user per slot and aligns the IRS phases to that user. The focus policy is round-robin during an initial warm-up period, then becomes adaptive.

\subsection{Sliding-window rate and priority weights}
Let
\begin{equation}
    \mathcal W(t)=\{\max(1,t-W+1),\dots,t\}
\end{equation}
denote the current averaging window. The sliding-window average rate of user $k$ is
\begin{equation}
    \bar R_k(t)=\frac{1}{\abs{\mathcal W(t)}}\sum_{\tau\in\mathcal W(t)}R_k(\tau).
    \label{eq:rbar}
\end{equation}
The priority weight is then defined as
\begin{equation}
    w_k(t)=\frac{1}{(\bar R_k(t)+\epsilon)^\beta},
    \label{eq:wk}
\end{equation}
where $\epsilon>0$ prevents division by zero and $\beta\ge 1$ controls how aggressively low-rate users are prioritized.

The normalized sampling probabilities are
\begin{equation}
    p_k(t)=\frac{w_k(t)}{\sum_{j=1}^{K}w_j(t)}.
    \label{eq:pk}
\end{equation}
After the warm-up period, the focus user for slot $t$ is sampled according to the distribution formed from the most recently available average rates, i.e.,
\begin{equation}
    k^{\star}(t)\sim\{p_k(t-1)\}_{k=1}^{K},
    \label{eq:focussample}
\end{equation}
with the understanding that, in implementation, the probabilities are updated from the rate history accumulated up to the preceding slot. In the reported run, $W=20$, $\beta=2$, and the first $W$ slots use round-robin focus assignment.

It should also be noted that the adaptive policy is fairness-aware by construction, since users with lower recent rates receive higher sampling probabilities. However, this work makes no claim of proportional-fair convergence, optimality, or guaranteed fairness attainment. The scheduler is used here as a lightweight adaptive mechanism rather than as a theorem-driven control law.

\begin{algorithm}[t]
\caption{Code-matched adaptive IRS scheduling and channel allocation}
\label{alg:main}
\begin{algorithmic}[1]
\Require Network parameters, IRS parameters, window $W$, exponent $\beta$, regularization $\epsilon$, sensing parameters $M$ and $P_{\mathrm{fa}}$, decode threshold $\gamma_{\mathrm{dec}}$.
\State Initialize geometry, random seed, node positions, node velocities, small-scale channels, rate histories, average-rate variables, and previous channel assignments.
\For{$t=1$ to $T$}
    \State Update user kinematics with reflective boundaries.
    \State Update direct-link and user-to-IRS small-scale channels via coherence-time redraws and Doppler phase evolution.
    \If{$t\le W$}
        \State Select the focus user by round robin.
    \Else
        \State Compute the most recently available sliding-window average rates, then form the weights and sampling probabilities according to \cref{eq:rbar}--\cref{eq:pk}.
        \State Sample the focus user according to \cref{eq:focussample}.
    \EndIf
    \State Compute the IRS phase profile for the selected focus user and quantize it to $2^b$ levels.
    \State Compute the direct, cascaded, and total channels, and then compute the received powers $P_k^{\mathrm{rx}}(t)$ for all users.
    \State Compute the exact energy-detection threshold $\gamma$ from \cref{eq:gammaexact}.
    \State Initialize the running channel energies according to \cref{eq:initE}.
    \For{$k=1$ to $K$}
        \State Assign $c_k(t)$ according to \cref{eq:alloc}.
        \State Update the corresponding running channel energy according to \cref{eq:eupdate}.
    \EndFor
    \State Compute each user SINR via \cref{eq:sinr} and the corresponding rate via \cref{eq:rate}.
    \State Set the rate to zero for users with $\mathrm{SINR}_k(t)<\gamma_{\mathrm{dec,lin}}$.
    \State Update the sliding-window average rates and priority variables for use in the next slot.
\EndFor
\State Output node-level rates, average SINRs, focus fractions, and network-level summary metrics.
\end{algorithmic}
\end{algorithm}

\section{Illustrative Seeded Numerical Example}
The simulation parameters used in the code-generated run are listed in \cref{tab:params}. The run uses fixed seed $42$ and spans $T=200$ slots, i.e., $1$ s with slot duration $\Delta t=5$ ms.

The MATLAB code reports average SINR in the following exact form:
\begin{equation}
    \mathrm{AvgSINR}_k = 10\log_{10}\!\left(\frac{1}{T}\sum_{t=1}^{T}\mathrm{SINR}_k(t)\right).
    \label{eq:avgsinr}
\end{equation}
Thus, the averaging is carried out in the linear domain and converted to dB afterward. This is important when comparing average SINR and average rate, because the rate is thresholded and depends on instantaneous channel sharing.

\begin{table}[H]
\centering
\caption{Simulation parameters for the demonstrative run.}
\label{tab:params}
\small
\begin{tabularx}{\linewidth}{@{}lX@{}}
\toprule
Parameter & Value \\
\midrule
Carrier frequency & $f_c=3.5$ GHz ($\lambda=c/f_c$) \\
Number of channels / bandwidth & $C=4$, $B=5$ MHz \\
Nodes / IRS size & $K=10$, $N=8\times 8=64$ \\
Element spacing / quantization & $\lambda/2$, $b=3$ bits (8 levels) \\
Reflection efficiency & $\rho=0.98$ \\
Path-loss exponent & $\alpha=2.2$ \\
Transmit power / noise figure & $P_{\mathrm{tx}}=20$ dBm, NF $=6$ dB \\
Noise power & $\sigma^2=k_B T_0 B\,10^{\mathrm{NF}/10}$, with $T_0=290$ K and $\mathrm{NF}=6$ dB \\
Decode threshold & $\gamma_{\mathrm{dec}}=-10$ dB \\
Sensing parameters & $M=128$, $P_{\mathrm{fa}}=0.1$ (exact $\chi^2$ threshold) \\
Mobility / slotting & $v_{\max}=3$ m/s, $\Delta t=5$ ms, $T=200$ \\
Scheduler parameters & $W=20$, $\beta=2$; warm-up: round robin for the first $W$ slots \\
BS / IRS-center positions & $\vect b=[0,0,10]^\top$ m, $\vect r_0=[30,0,8]^\top$ m \\
Initialization region / motion & $\mathcal A=[-50,50]\times[-50,50]\times[0,3]$ m; random planar heading, speed $\le v_{\max}$, reflective boundaries \\
\bottomrule
\end{tabularx}
\end{table}

\subsection{Per-node statistics}
\Cref{tab:nodes} reproduces the per-node average SINR, average rate, and IRS focus fraction reported by the code for the seeded adaptive run. The results should be interpreted with caution: the scheduler reallocates focus time to several lower-rate users, but the resulting operating point remains strongly heterogeneous across the network.

\begin{table}[H]
\centering
\caption{Per-node statistics for the adaptive run.}
\label{tab:nodes}
\small
\begin{tabularx}{\linewidth}{@{}cS[table-format=2.2]S[table-format=2.2]S[table-format=2.1]X@{}}
\toprule
Node & {Avg SINR} & {Avg Rate} & {IRS Focus} & Comment \\
 & {(dB)} & {(Mbps)} & {(\%)} & \\
\midrule
1  & -8.78 & 0.68 & 25.0 & low-rate and heavily prioritized \\
2  & -0.20 & 1.87 & 8.0  & moderate \\
3  & 21.47 & 1.37 & 10.0 & favorable SINR but intermittent throughput \\
4  & -10.79 & 0.35 & 12.5 & below the decode threshold on average \\
5  & -8.97 & 0.55 & 14.0 & low-rate and prioritized \\
6  & -6.38 & 1.08 & 10.5 & low-to-moderate \\
7  & 32.71 & 6.70 & 1.5  & strong \\
8  & 31.50 & 4.24 & 16.5 & strong, but still given notable focus time \\
9  & 30.91 & 5.94 & 1.0  & strong \\
10 & 43.43 & 19.69 & 1.0 & dominant \\
\bottomrule
\end{tabularx}
\end{table}

\begin{figure}[H]
    \centering
    \includegraphics[width=\linewidth]{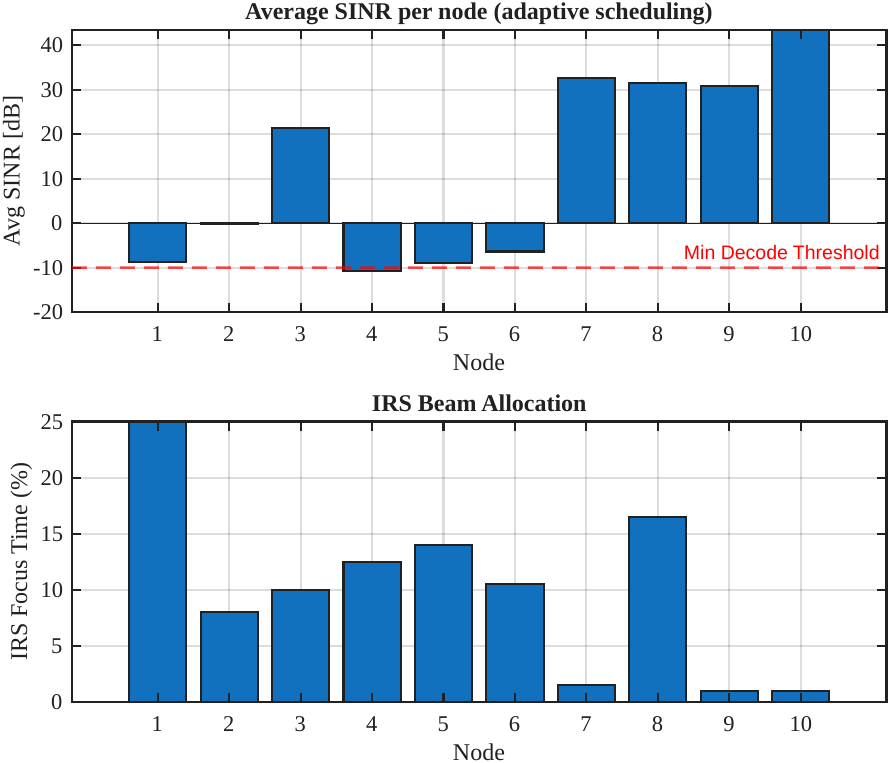}
    \caption{Per-node average SINR together with the minimum decode threshold, as well as the corresponding IRS focus-time allocation for the seeded adaptive run.}
    \label{fig:sinr_focus}
\end{figure}

\subsection{Network-level behavior}
\Cref{tab:netmetrics} summarizes the network-level outcomes of the adaptive run. The reported fairness figures are presented without optimistic interpretation: although the scheduling rule is fairness-aware by construction, the resulting operating point remains highly unequal, and Node~4 stays below the decode threshold on average. \Cref{fig:timeseries} shows the corresponding temporal sum-rate and per-node rate trajectories. These curves are useful as implementation-faithful diagnostics of the seeded run, but they should not be interpreted as a statistically complete performance characterization.

\begin{table}[H]
\centering
\caption{Network-level metrics for the seeded adaptive run.}
\label{tab:netmetrics}
\small
\begin{tabularx}{0.82\linewidth}{@{}lX@{}}
\toprule
Metric & Value \\
\midrule
Average sum rate & 42.47 Mbps \\
Jain's fairness index & 0.366 \\
Min/max rate ratio & 0.018 \\
Nodes below decode threshold on average ($-10$ dB) & 1 (Node 4) \\
\bottomrule
\end{tabularx}
\end{table}

\begin{figure}[H]
    \centering
    \includegraphics[width=\linewidth]{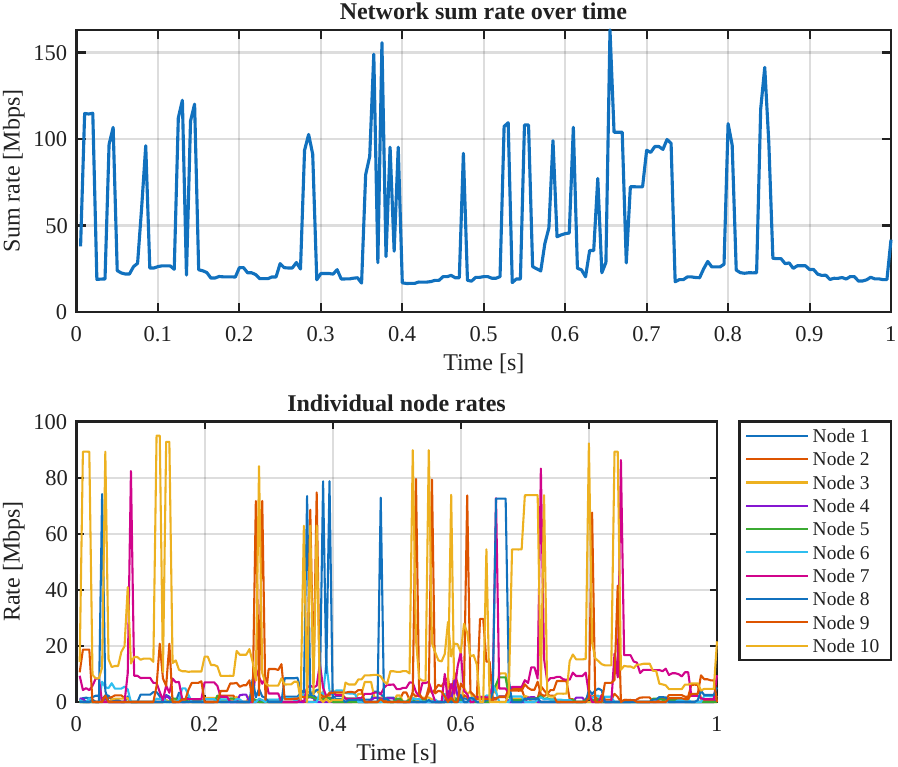}
    \caption{Temporal trajectories for the adaptive run: total network sum rate (top) and individual rates (bottom). It depicts strong temporal variability and persistent heterogeneity across users.}
    \label{fig:timeseries}
\end{figure}

The example should be viewed as a foundational demonstration of the proposed framework rather than a statistically complete performance characterization. In particular, it should not be used to claim broad success in fairness. Further study would require multi-seed averaging, sensitivity sweeps, and explicit baseline comparisons under different settings. It should also be noted that although the specific user identities and numerical values vary with the random seed, the qualitative behavior is expected to remain similar across runs: weaker users tend to receive higher IRS focus probability, while performance heterogeneity may still persist.

\section{Conclusion}
This paper has presented a self-contained, implementation-aligned framework for a simplified, mobility-aware, IRS-assisted multi-node uplink. In this adopted model, the framework combines direct and IRS-assisted propagation, geometric IRS phase control with finite-bit quantization, inverse-rate adaptive IRS-focus selection, and sequential energy-guided channel allocation. The analytical development has been deliberately restricted to statements that are supported by the stated assumptions and by the implemented simulator, namely the physics-based two-hop cascaded path-loss structure, the expectation-level reflected-power characterization under explicit independence assumptions, and the exact chi-square threshold used for energy detection.

In this form, the value of the work lies in providing a coherent baseline that is mathematically consistent within its own scope, transparent in its assumptions, and suitable for further refinement. The numerical example shows that the adaptive focusing rule is fairness-aware by construction, since weaker users tend to receive more IRS focus time, but it also confirms that performance disparity can persist, thereby serving as a foundational reference for subsequent extensions.

Future directions can address the limitations identified in this work. First, the study can be expanded beyond a single-seeded run to include confidence intervals and comparisons against benchmark scheduling and other low-complexity control policies. Second, the mobility and fading models can be strengthened by richer temporal correlations, more realistic CSI acquisition, and control-delay effects, particularly for dynamically reconfigured IRSs. Third, the scheduling and channel-allocation methods can be generalized toward learning-based designs, backed by a more systematic fairness-throughput tradeoff analysis under a wider range of deployment conditions. Lastly, broader validation across network densities, IRS sizes, resolutions, and sensing settings would help distinguish outcomes that are more robust in specific propagation scenarios.

\end{document}